\documentclass[twoside]{article}

\usepackage[sc]{mathpazo}
\usepackage[T1]{fontenc}
\usepackage{microtype}
\usepackage{amsmath}
\usepackage{amsfonts}
\usepackage{amsthm}

\theoremstyle{plain}
\newtheorem{thm}{Theorem}[subsection]
\newtheorem{lma}[thm]{Lemma}

\newtheorem{obs}[thm]{Observation}

\theoremstyle{definition}

\newcommand{\COL}[1]{$(2,#1)$-$\mathcal{COL}$}
\newcommand{\SAT}{$\mathcal{SAT}$}
\newcommand{\NAESAT}{$\mathcal{NAE}$-$\mathcal{SAT}$}
\newcommand{\mCOL}[1]{(2,#1)\text{-}\mathcal{COL}}

\newcommand{\mNAESAT}{\mathcal{NAE}\text{-}\mathcal{SAT}}

\usepackage[top=32mm,left=18mm,right=18mm,bottom=22mm,columnsep=20pt]{geometry}
\usepackage{multicol}
\usepackage[hang, small,labelfont=bf,up,textfont=it,up]{caption}
\usepackage{booktabs}
\usepackage{float} 

\usepackage{paralist} 
\usepackage{graphicx}

\usepackage{abstract}

\usepackage{titlesec}
\renewcommand\thesection{\Roman{section}}
\titleformat{\section}[block]{\large\scshape\centering}{\thesection.}{1em}{}
\titleformat{\subsection}[block]{\large}{\thesubsection.}{1em}{}

\usepackage{fancyhdr}
\usepackage{algpseudocode}

\usepackage{framed}

\title{\vspace{-15mm}\fontsize{24pt}{10pt}\selectfont\textbf{Vertex 2-coloring without monochromatic cycles}}

\author{
\large
\textsc{Micha\l~Karpi\'nski}\thanks{PhD student (at the time of writing)}\\[2mm]
\normalsize University of Wroclaw \\
\normalsize karp@cs.uni.wroc.pl
\vspace{-5mm}
}
\date{}

\begin{document}

\maketitle 

\thispagestyle{fancy} 

\begin{abstract}

  \noindent In this paper we study a problem of vertex two-coloring of undirected graph such that
  there is no monochromatic cycle of given length. We show that this problem is hard to solve. We give a proof
  by presenting a reduction from variation of satisfiability (SAT) problem. We show nice properties of coloring cliques with two colors
  which plays pivotal role in the reduction construction. 
\end{abstract}

\begin{multicols}{2} 


  \section{Introduction}

  Vertex coloring problems (VCP) have been studied extensively since the inception of graph theory.
  In classical form, problem of $k$-coloring a graph is stated like this: can we color vertices of a graph using $k$
  different colors, so that no neighbouring vertices have the same color? It is known that this problem is NP-complete \cite{vcp1}.
  VCPs have received much attention in the literature not only for its theoretical aspects and difficulty from the computational
  point of view, but also for its real world applications, for example in: scheduling \cite{vcp2},
  timetabling \cite{vcp3}, register allocation \cite{vcp4}, train platforming \cite{vcp5},
  frequency assignment \cite{vcp6}, communication networks \cite{vcp7} and many other engineering fields.

  In this paper we study a variation of the coloring problem. Using only two colors we want to color the vertices, so that
  there is no monochromatic cycle of given length. There have been some research in solving a slightly different problem: is there
  a 2-coloring such that there exists no monochromatic cycles (of any length). This problem can be viewed as partitioning a graph into
  two induces forests and it is known to be NP-complete \cite{2col1} for directed graphs.
  Another result worth mentioning is by Nobinon et al. \cite{2col2} where authors show that this problem is NP-complete
  even for oriented graphs. They also give implementation of three exact algorithms and some inapproximability results. The motivation
  to study this class of problems lies in economics -- 2-coloring without monochromatic cycles can be used in the study of rationality
  of consumption behavior.

  Many more papers have been written on subject of acyclic coloring (or partitioning). Papers relevant to ours include (among many others):
  \cite{2col3}, \cite{2col4}, \cite{2col5}.
  
  The rest of the paper is organized in the following way: in section 2 we define notation used in this paper, we also give
  definitions of studied problems and we state the main theorem. In section 3 and 4 we prove the hardess of our coloring problem
  for cycles of small length (3, and 4). Later, in section 5 we generalize the ideas used in previous sections to prove the main theorem.
  We end the paper with some conclusions and we show perspectives for future work.
  
  \section{Preliminaries}

  The purpose of this section is to introduce reader to notation used in later chapters as
  well as definitions of studied problems. Let $G=(V,E)$ be an undirected, unweighted graph.
  The {\em cycle} in $G$ is a vertex disjoint, closed, simple path in $G$. We denote $\mathcal{C}_k$ to be a set of
  all cycles in $G$ of length $k$.
  Let $c: V \rightarrow \{r,b\}$ be a mapping that for each vertex in $V$ assigns one of two colors ({\em red} or {\em blue}).
  We will call any such $c$: the {\em coloring} of graph $G$. Furthermore, we will say
  that given coloring $c$ is {\em valid}, if a certian predicate $P(c)$ is true. Let $K_n$ be a clique of size $n$, that is: a
  graph with $n$ vertices in which every vertex is connected by an edge to any other vertex.

  Let \COL{k} be the decision problem of whether there exists a valid 2-coloring for given graph.
  We give the validity predicate $P_k(c)$ below. It is true only if the coloring $c$ does not contain any cycles of size $k$ with
  vertices of the same color.

  \[
    P_k(c) \equiv \forall Q \in \mathcal{C}_k \exists u,v \in Q \quad c(u) \neq c(v) 
  \]
  
  Formally, our problem can be expressed as:

  \[
    (2,k)\text{-}\mathcal{COL} = \{G \, : \, \exists c \, P_k(c)\}
  \]

  We are interested in knowing how hard is the question, whether given graph $G$ belongs to \COL{k}.
  In the next two sections we study the simplest variants, that is when $k=3$ and $k=4$. Cycle of size three we
  call a {\em triangle}, and of size four: a {\em square}.

  Let \SAT denote the classical boolean satisfiability problem. Namely, it is the set of all boolean formulas
  in CNF (conjunctive normal form)
  for which there exists a truth assignment that satisfies it. It is known that this problem is NP-complete \cite{cook}.
  It is also known that
  a certain variation of \SAT called \NAESAT (not-all-equal SAT) is NP-complete \cite{naesat}. In this variation we impose additional
  constraint on the satisfing assignment: each clause has at least one literal that is true, and at least one that is false.
  We denote $k$-\SAT and $k$-\NAESAT (for $k \geq 3$) to be subsets of \SAT and \NAESAT where each clause in given formula has at most $k$
  literals (it's in kCNF). For $k < 3$ for both problems there exists polynomial time algorithms that solves them.

  We are ready to state the main theorem:

  \begin{thm}
    \label{thm:main}
    For any integer $k \geq 3$, \COL{k} is NP-complete.
  \end{thm} 

  In order to prove theorem \ref{thm:main}, we will prove the following theorem:

  \begin{thm}
    \label{thm:red}
    For any integer $k \geq 3$, there exists a computable function $f$, such that
    for any boolean formula $\phi$, $\phi \in k$-$\mathcal{NAE}$-$\mathcal{SAT}$ if and only if 
    $f(\phi) \in$\COL{k}.
  \end{thm}
  
  \section{Two-coloring without monochromatic triangles}
  
  In this section we prove theorem \ref{thm:red} for $k=3$. Let $\phi$ be
  a boolean formula in 3CNF with $n$ variables $x_1,\ldots,x_n$ and $m$ cluses
  $C_1,\ldots,C_m$. We construct desired graph $G_{\phi}$ in the following way.
  Let us begin by showing an abstract form of $G_{\phi}$. The reduction consists of three gadgets: one for
  each variable, one for each clause, and one for each {\em super-edge}. The super-edge $\{u,v\}$ is an edge with a property, that
  any valid coloring $c$ implies that $c(u) \neq c(v)$. For starters, assume that we already have such edges at our disposal.
  This is how we would construct $G_{\phi}$: a gadget for variable $x$ consists of two vertices labeled $x$ and $\neg x$ connected by
  a super-edge. Gadget for clause $C=(u \vee v \vee w)$ consists of a triangle with vertices labeled $u,v$ and $w$. We connect each literal
  from variable gadget to its every occurrence in clause gadgets using super-edges. Example is given in figure \ref{fig:basic_red} for
  formula $\phi=(x_1 \vee \neg x_1 \vee x_2) \wedge ( \neg x_2 \vee x_3 \vee \neg x_3)$. Dashed lines represent super-edges. We prove
  that this is indeed the correct reduction.

  \begin{figure}[H]
    \begin{center}
      \includegraphics[scale=0.5]{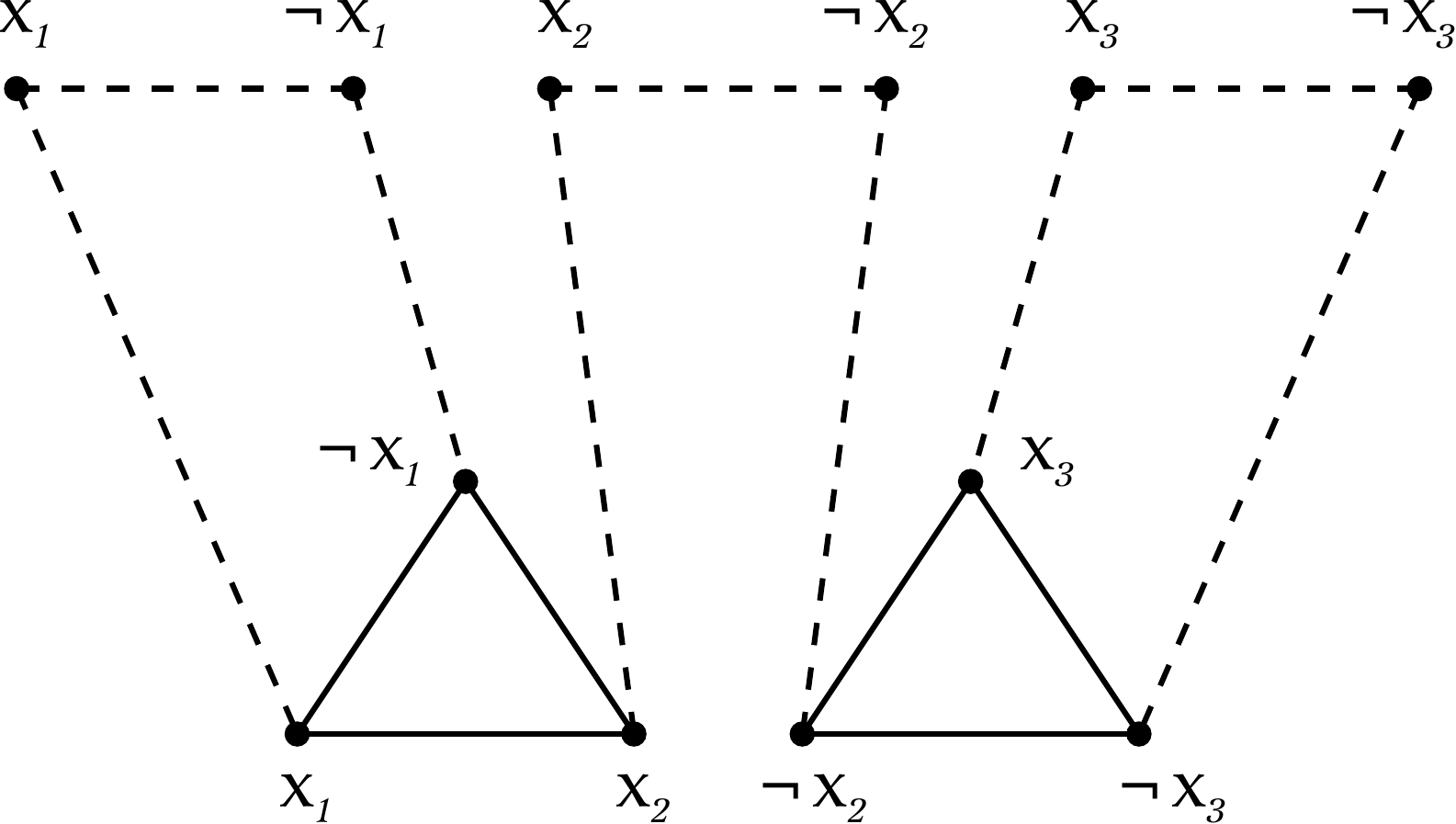}
    \end{center}
    \caption{Abstract form of graph used in reduction}
    \label{fig:basic_red}
  \end{figure}
  
  \begin{lma}
    For any given $\phi$, graph $G_{\phi}$ has a property, that:

    \[
      \phi \in 3\text{-}\mNAESAT \iff G_{\phi} \in \mCOL{3} 
    \]
  \end{lma}

  \begin{proof}
    First we assume that $\phi \in 3$-\NAESAT and let $\hat{\sigma}(x_1,\ldots,x_n)$ be the truth assignment that certify it.
    Each vertex with non-negated label $x$ in vertex gadgets is colored red if $\sigma(x)=T$ and blue otherwise. Coloring of every other
    vertex is forced by super-edges. Notice that the only place where there could be any monochromatic triangle is in some clause gadget.
    We cannot make that trinagle using mixture of vertices from other clause gadgets or vertex gadgets because we always have to pass
    through a super-edge, hence we change a color of vertices on our path. Now if we assume on the contrary, that some clause
    $C=(u \vee v \vee w)$ form a monochromatic triangle, then either $\sigma(u)=\sigma(v)=\sigma(w)=T$
    or $\sigma(u)=\sigma(v)=\sigma(w)=F$, which gives a contradiction.

    Now let $c$ be the valid coloring of $G_{\phi}$. Since $G_{\phi}$ has no monochromatic triangles, and from the property of super-edge
    we simply assign value $T$ for all variables from vertex gadgets that have color red, and $F$ otherwise. This
    gives an assignment $\hat{\sigma}(x_1,\ldots,x_n)$ that proves that $\phi \in 3$-\NAESAT. To see that, observe that
    every clause corresponding to clause gadget will have at least one literal that is true, and at least one that is false,
    because this clause gadget does not form a monochromatic triangle, which was assumed.
  \end{proof}

  All we have to do now is construct a gadget for super-edge. Such gadget need to have a property, that some
  selected edge $\{x,y\}$ in that gadget will always have $c(x)\neq c(y)$, for any valid coloring $c$ of that gadget
  (a valid coloring also has to exist). An example of the gadget is shown in figure \ref{fig:super_edge1}. On the left picture
  edge $\{x,y\}$ is pointed out. In the middle we have an example of valid coloring, and on the picture on the right we see how
  coloring $\{x,y\}$ in one color gives a contradiction (vertex with a question mark cannot be colored neither red, nor blue).
  The existence of this gadget completes the proof of theorem \ref{thm:red} for $k=3$ (and also theorem \ref{thm:main},
  with additional observation that our reduction is polynomial with respect to size of $\phi$).

  We argue, that even if a super-edge in figure \ref{fig:super_edge1} is enough to verify the
  genuineness of theorem \ref{thm:red} (for $k=3$), it is not ellegant.
  We give a better construction of the gadget that uses a certain coloring property
  of $K_4$. Our method is also easier to generalize for $k > 3$.

  The basic observation is that when we color any two vertices of $K_4$ in one selected color -- lets say red -- then
  the other two vertices will have to be colored blue (otherwise there would be a monochromatic triangle).
  Now if we were to {\em hook} another $K_4$ to those blue vertices (see figure \ref{fig:coloring_property}) then the two
  non-colored vertices would have to be red, and so on, and so on. With this we can create {\em strings} of $K_4$-s.

  \begin{figure}[H]
    \begin{center}
      \includegraphics[scale=0.7]{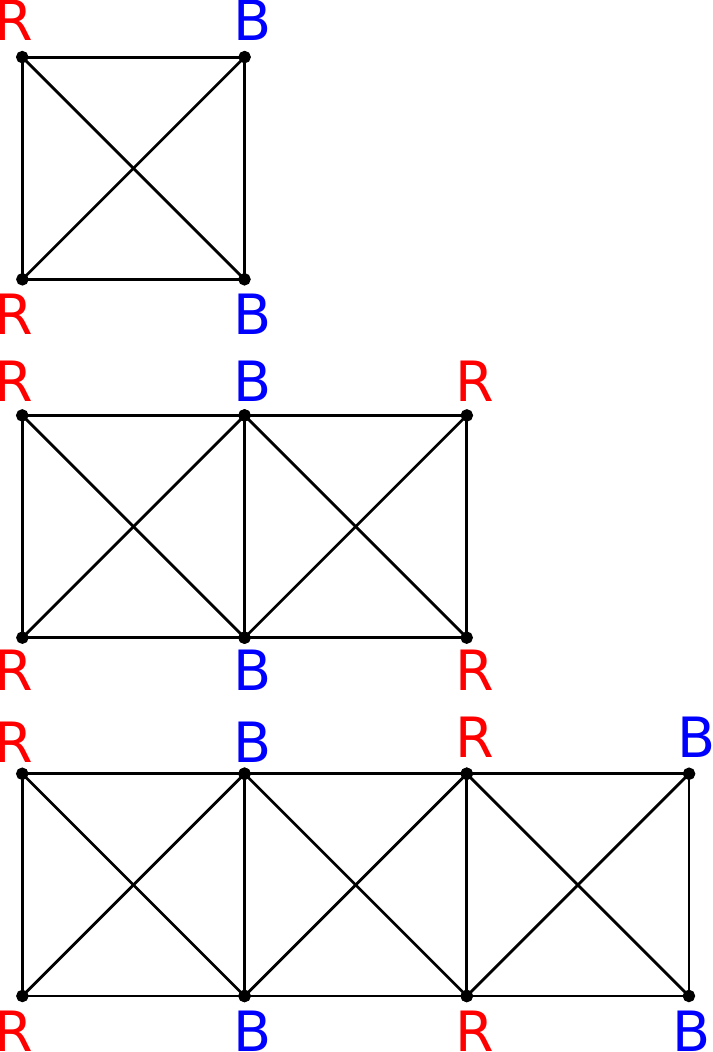}
    \end{center}
    \caption{The {\em strings} created from joining 4-cliques}
    \label{fig:coloring_property}
  \end{figure}

  The trick is to tie two ends of the string together. This will form a {\em loop}.
  It is easy to verify, that loop of length 5 is the desired gadget for super-edge.
  We show its properties in figure \ref{fig:super_edge2}. On the left, the edge $\{x,y\}$
  is pointed out. In the middle picture we give some valid coloring, and the last picture shows how
  coloring $\{x,y\}$ in a single color leads to a contradiction (vertices marked in question marks cannot be colored
  without creating a monochromatic triangle). In fact, we can make an easy observation:

  \begin{obs}
    Any loop of odd length (for lengths greater than 3) can be used as a gadget for super-edge.
  \end{obs}
 
  It turns out that loop of length 3 cannot be used, because it is isomorphic to $K_6$ (and therefore is
  not colorable). Also, coloring loops of even length would not lead to a contradiction, no matter which edge you choose for $\{x,y\}$.
  We leave verification of this statements to the reader. 

  Our symmetric gadget is slightly bigger than the one in figure \ref{fig:super_edge1}. It's
  25 edges and 10 vertices versus 21 edges and 9 vertices. It has been computed (by brute-force), that there
  is no super-edge gadget that uses 8 vertices or less. We did not bother
  to check if there is a gadget with number of edges less than 21. Using the symmetric
  gadget we can compute number of edges and vertices used in entire $G_{\phi}$. Let $E_c$, $E_v$, $E_s$
  denote a set of edges used in all clause gadgets, all vertex gadgets and all super-edge gadgets respectively.
  We define $V_c$, $V_v$, $V_s$ in a similiar fashion. We have:

  \begin{align*}
    |E(G_{\phi})| &= |E_c| + |E_v| + |E_s| \\
    &= 3m + 0 + 25(3m+n) \\
    &= 78m+25n
  \end{align*}

  \begin{align*}
    |V(G_{\phi})| &= |V_c| + |V_v| + |V_s| \\
    &= 3m + 2n + (10-2)(3m+n) \\
    &= 24m+10n
   \end{align*}

  This shows that reduction can be performed in polynomial time (with respect to $n$ and $m$) and therefore
  completes (yet another) proof of theorem \ref{thm:red}. But we can improve the reduction even further and
  push properties of our symmetric gadget to its limit. 

  We will now show what we call {\em The Necklace Reduction}. If we look at a loop of size $l$, we will spot as many as $l$
  candidates for chosing the edge $\{x,y\}$. This is easily seen in figure 5. The symmetry of our gadget guarantees,
  that any edge on the juncture of $K_4$-s can be considered ${x,y}$. But that leaves $l-1$ candidates unused. In nacklace reduction we
  get rid of wasting so many useful edges (to some extent). We simply weave all vertex gadgets on a single loop of length $2n+1$.
  Vertex gadget for variable $x_i$ (for $i=1..n$) now becomes egde on the juncture of $(2i)$-th and $(2i+1)$-th $K_4$-s (numeration
  can start at any arbitrary $K_4$). We leave the rest of reduction the same as before. We have now created a beautiful necklace
  of which example can be seen in figure \ref{fig:necklace} (it uses formula from previous example; some labels were omitted).

  \begin{figure}[H]
    \begin{center}
      \includegraphics[scale=1.0]{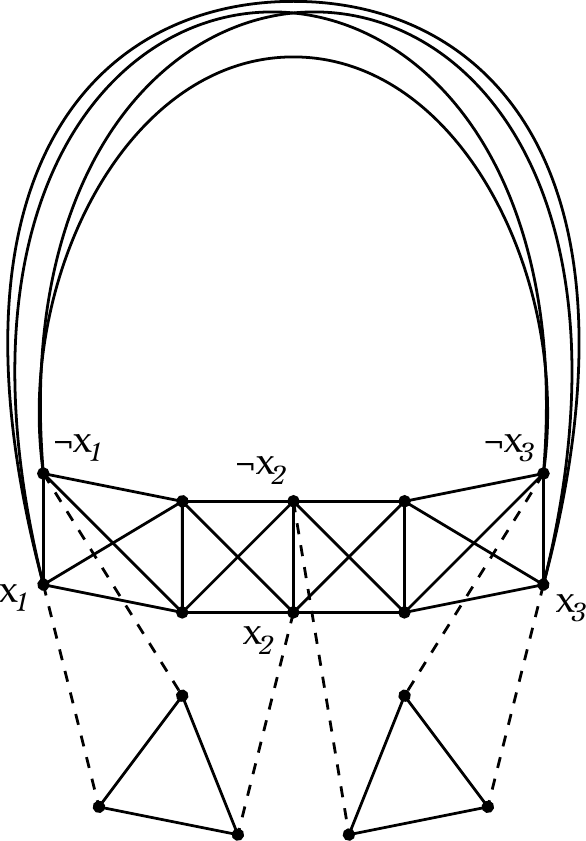}
    \end{center}
    \caption{The Necklace Reduction}
    \label{fig:necklace}
  \end{figure}

  Number of edges and vertices drops down to:

  \[
  |E(G_{\phi})|= 78m + 10n + 5, \quad |V(G_{\phi})|= 27m + 4n + 2
  \]

  We can further improve the necklace by weaving all other super-edges, but
  the construction is rather complicated. Details will be available in extended
  version of this paper.
  
  \begin{figure*}[t]
    \begin{center}
      \includegraphics[scale=0.6]{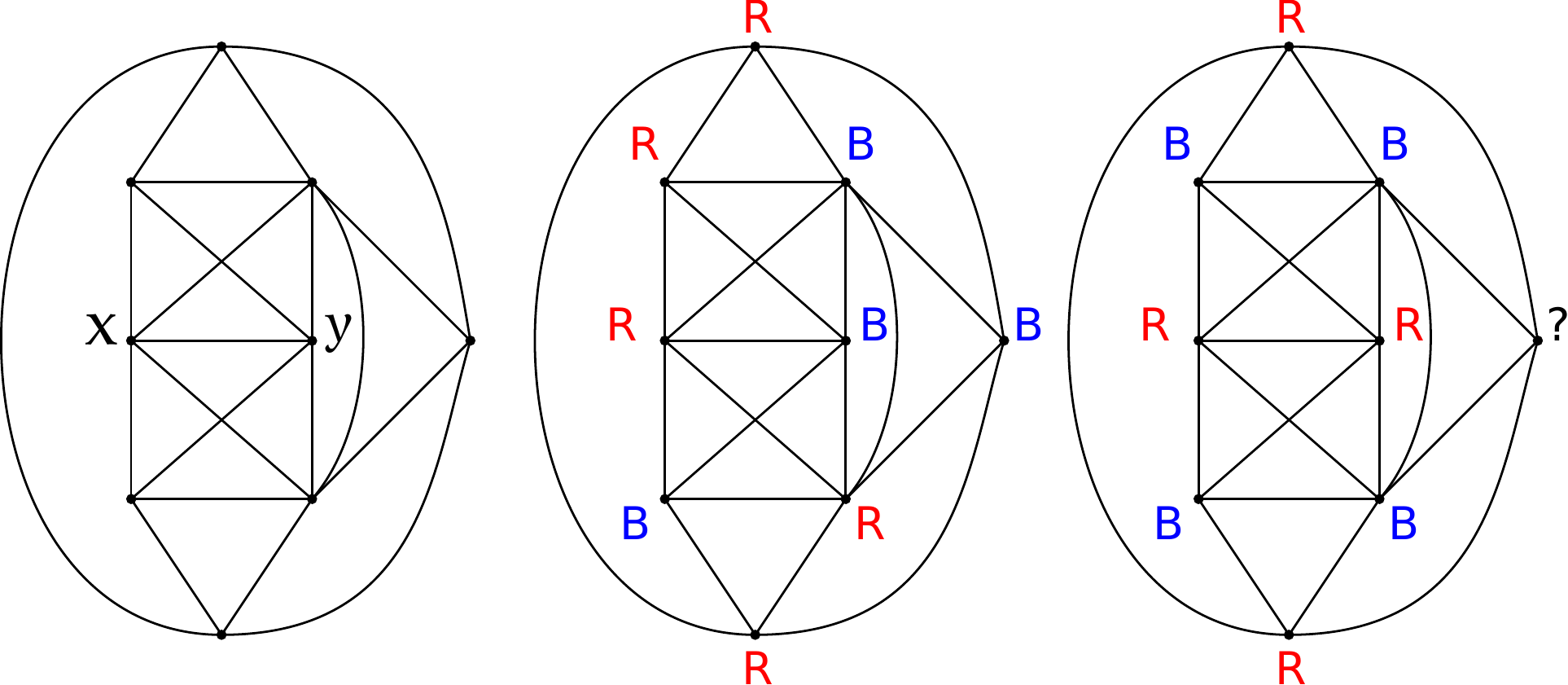}
    \end{center}
    \caption{The super-edge gadget}
    \label{fig:super_edge1}
  \end{figure*}

  \begin{figure*}[t]
    \begin{center}
      \includegraphics[scale=0.6]{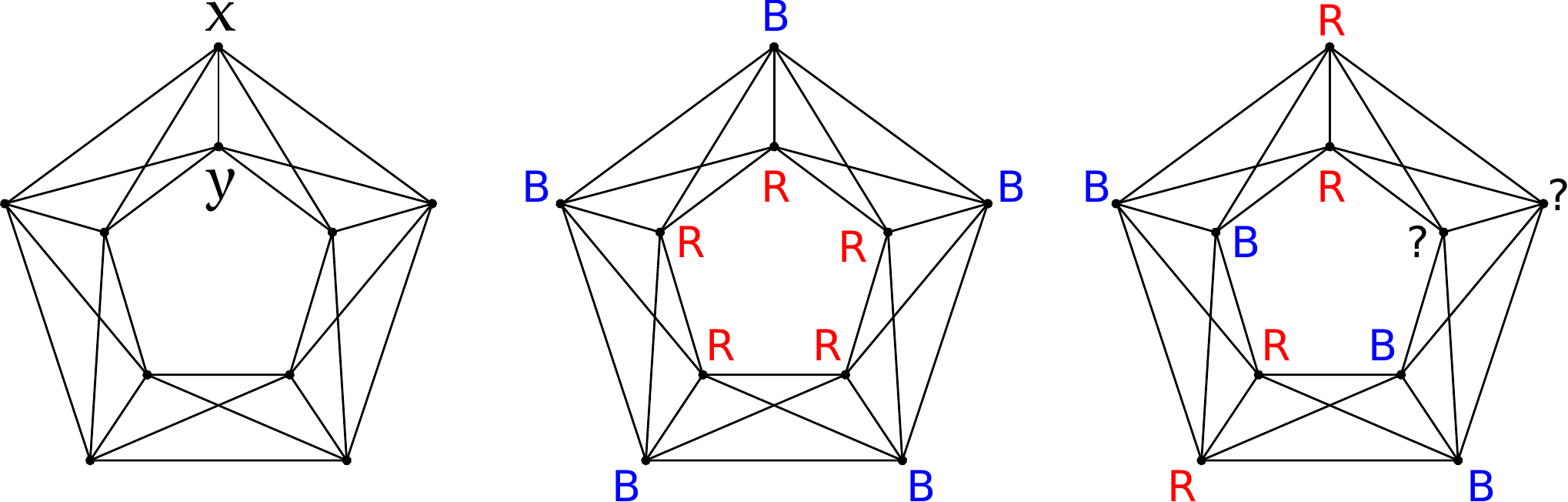}
    \end{center}
    \caption{The symetric super-edge gadget}
    \label{fig:super_edge2}
  \end{figure*}
  
  \section{Two-coloring without monochromatic squares}

  In this section we extend our reduction to cycles of length 4.
  The abstract form of $G_{\phi}$ remains almost the same and the only
  diffrence is that we have squares in place of triangles for clause
  gadgets. In fact we use the similiar graph for higher values of $k$.
  Proof of correctness is the same as before, so we leave the details
  to the reader.

  The most importnant part is to construct a gadget for super-edges. Now,
  we want to create a graph with a selected edge $\{x,y\}$ that there exists
  a valid coloring (without monochromatic squares) and that in every valid coloring
  $c$: $c(x) \neq c(y)$. We use $K_6$ as a building block for the gadget and exploit
  its coloring property.

  \begin{figure}[H]
    \begin{center}
      \includegraphics[scale=0.6]{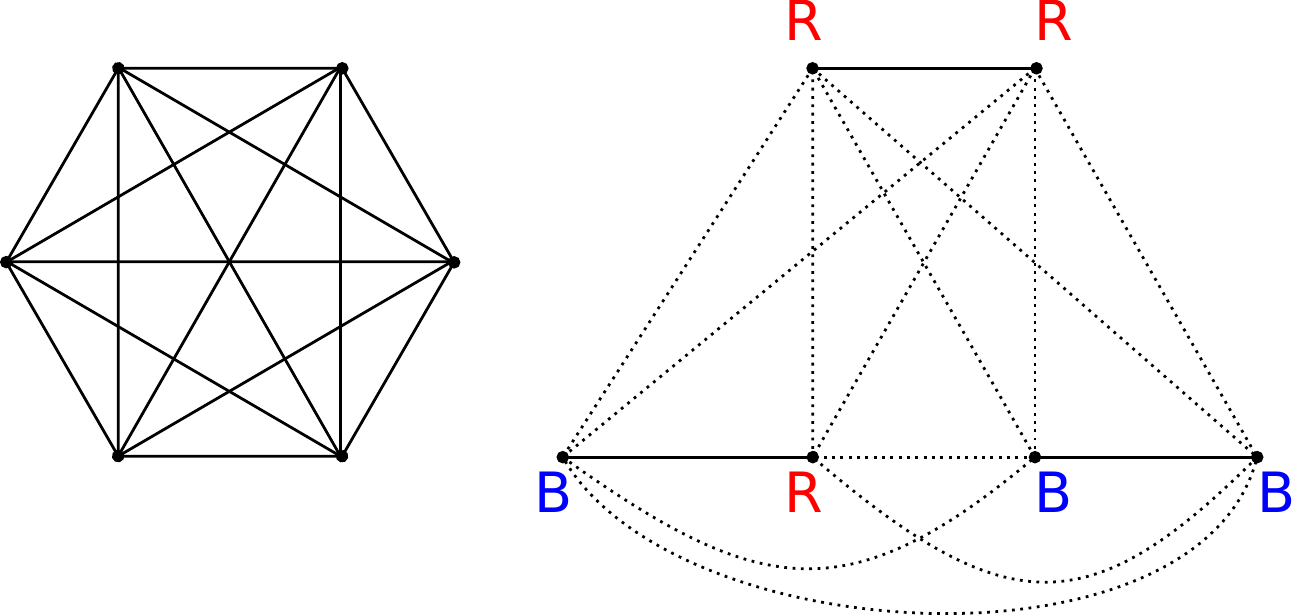}
    \end{center}
    \caption{Coloring property of 6-clique}
    \label{fig:coloring_property_k6}
  \end{figure}

  In figure \ref{fig:coloring_property_k6} on the left we see $K_6$. On the right it is the same $K_6$, but
  with rearranged edges. Three arbitrarily chosen, disjoint edges have been pointed out and streched in three different directions.
  Rest of the edges are less significant so we placed dotted lines in their place. Notice, that when we color vertices of top edge
  in a single color -- let's say red -- then by using easy pigeon hole argument we can conclude, that exactly one of two bottom edges
  will have both of its vertices colored blue.
  
  To further simplify the $K_6$, imagine that the selected edges become nodes and that there are lines between top node
  and two bottom nodes. This creates a reverse v-shaped component. The node which has two different colors associated
  to it, we label as $X$ (see figure \ref{fig:coloring_property_k6_2}).

  \begin{figure}[H]
    \begin{center}
      \includegraphics[scale=1.0]{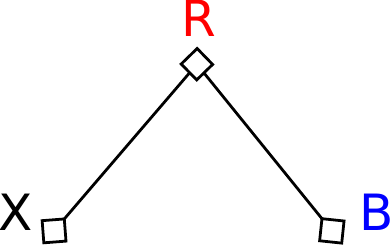}
    \end{center}
    \caption{Simplifing $K_6$}
    \label{fig:coloring_property_k6_2}
  \end{figure}

  Now we present the trick to our gadget. We build a full binary tree of height 4, consisting of reverse v-shaped components.
  It follows from coloring property of $K_6$ discussed before, that if we color root node in red, then there exist a path
  from root to leaf with alternating colors (see figure \ref{fig:super_edge_sq}). Notice the analogy to the construction
  of strings in previous section.

  \begin{figure}[H]
    \begin{center}
      \includegraphics[scale=1.0]{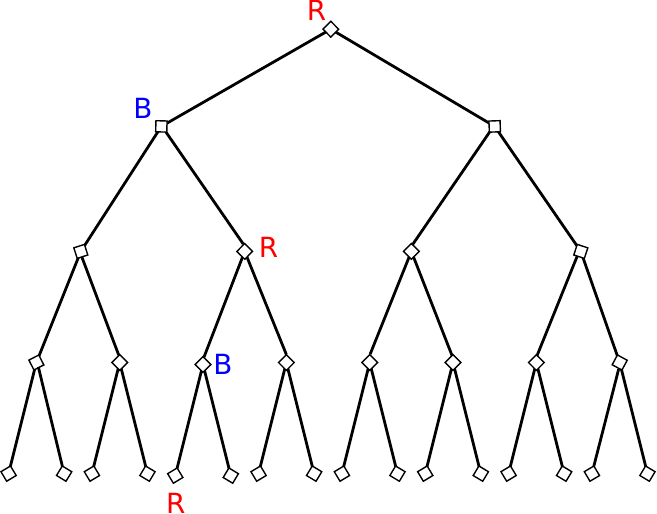}
    \end{center}
    \caption{Super-edge gadget. All leafs are connected to root.}
    \label{fig:super_edge_sq}
  \end{figure}

  To achieve a contradition we connect all leafs to the root using two edges for each leaf in a way that they form a square.
  This completes the construction. We choose root node as $\{x,y\}$.

  Chosing the height 4 for $T$ is not a coincidence, as using any tree of smaller size would either not lead to contradiction 
  (heights 1 or 3) or would not be colorable -- for height 2 we can find a monochromatic square in any coloring. We again
  leave verification to the reader.

  It remains to show that our gadget has a valid coloring. We simply label all nodes by $X$. We now prove that this will
  not create any monochromatic square. There are two places in our gadget that require special attention:

  \begin{itemize}
    \item $\mathcal{P}1$. Connections between inner nodes of the tree, and
    \item $\mathcal{P}2$. Connections between leafs and root.
  \end{itemize}

  Both of them can be handaled in a strightforward way. For the former look at figure \ref{fig:connections_1},
  where we reverse the process of
  $K_6$-simplification for some subtree of $T$. We quickly verify, that there are no monochromatic squares. This is the smallest,
  nontrivial subtree in which there could lurk some hidden monochromatic squares. Thanks to regular structure and symmerty of full
  binary trees, any other combinations of nodes need not be checked. One could use induction for formal proof, but we will leave
  it like this.

  \begin{figure}[H]
    \begin{center}
      \includegraphics[scale=1.0]{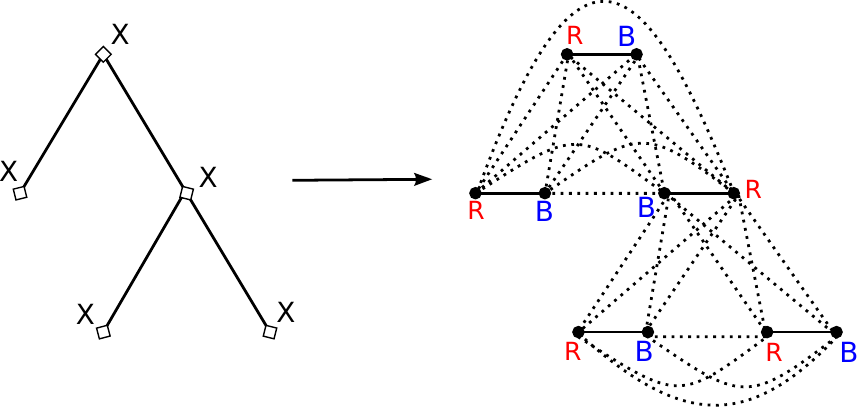}
    \end{center}
    \caption{Connections between inner nodes expanded.}
    \label{fig:connections_1}
  \end{figure}

  $\mathcal{P}2$ couses some minor troubles. Take a look at figure \ref{fig:connections_2}.
  Notice that we found a monochromatic square. This leads
  to conclusion that not every coloring that labels each node by $X$ is valid. We can fix that by coloring both leafs so that
  they form alternating squares with the root (the color is alternating). Now any path that passes from leaf to root have to change
  the color, so there are no more threat to find a monochromatic square.

  \begin{figure}[H]
    \begin{center}
      \includegraphics[scale=1.0]{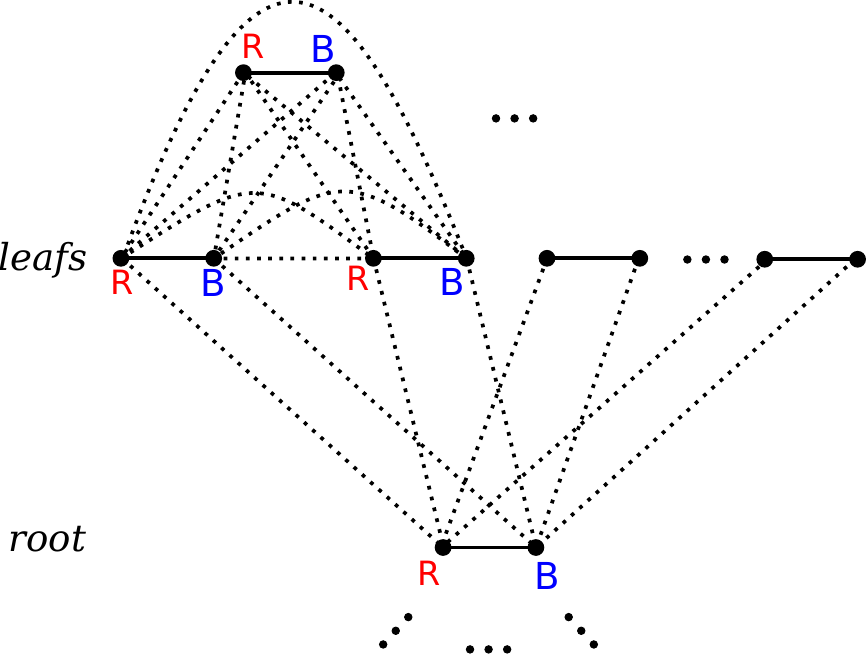}
    \end{center}
    \caption{Connections between leafs and root.}
    \label{fig:connections_2}
  \end{figure}

  This completes the proof of theorem \ref{thm:red} for $k=4$.
  We see that this is a polynomial reduction. For sake of completness lets count
  number of edges and vertices in a single super-edge gadget, and then in entire graph $G_{\phi}$:

  \[
  \text{\#edges-in-gadget} = 243, \quad \text{\#vertices-in-gadget} = 62
  \]

  \begin{align*}
    |E(G_{\phi})| &= |E_c| + |E_v| + |E_s| \\
    &= 4m + 0 + 243(4m+n) \\
    &= 976m + 243n
  \end{align*}

  \begin{align*}
    |V(G_{\phi})| &= |V_c| + |V_v| + |V_s| \\
    &= 4m + 2n + (62-2)(4m+n) \\
    &= 244m + 62n
  \end{align*}

  \section{The general case}

  In this section we finally prove theorem \ref{thm:red} for $k>4$. We do this by expanding the binary tree gadget from last section.
  The tree will grow exponantially with respect to $k$, but remember that $k$ is a constant associated with the problem \COL{k}, so 
  our reduction will still be polynomial in size of $\phi$ (but very, very big). Our goal now is to construct a graph with a selected edge
  $\{x,y\}$, that there exists a valid coloring (without monochromatic cycles of length $k$) and that in every valid coloring $c$:
  $c(x) \neq c(y)$. For now assume that $k$ is even. This will simplify our reasoning.

  First we construct a binary tree $T$ consisting of reverse v-shaped components introduced in previous section. Let
  height of $T$ be $h=4 \lfloor \frac{k-1}{2} \rfloor$. For $i=1..\lfloor \frac{k-1}{2} \rfloor$, we will call all nodes of
  depth $4i$: {\em cycle-inducing} (notice that root and leafs are also cycle-inducing).
  Let $CI$ be the set of all cycle-inducing nodes in $T$. If we color root node in a single
  color -- let's say red -- then there exists a path $P$ from root to some leaf, with alternating colors. Notice that all nodes in
  $P \cap CI$ are now colored red. Those nodes will create a monochromatic cycle of length $k$. To achieve this, we add edges between
  cycle-induced nodes in the following way.

  First, we connect root and leafs just like in previous section. Next, for each cycle-induced
  node of depth $4i$ ($i=1..(\lfloor \frac{k-1}{2} \rfloor-1)$) we conect it to all its descendants on depth $4(i+1)$ (they also belong
  to $CI$). We add edges between them the same way we did with root and leafs. The example of how this produces monochromatic cycle is
  shown on in figure \ref{fig:super_edge_general}. If we take the graph induced by $P \cap CI$, it forms a {\em donut} shown
  in the right picture. We can easily identify a monochromatic cycle of length $k$.

  \begin{figure}[H]
    \begin{center}
      \includegraphics[scale=1.0]{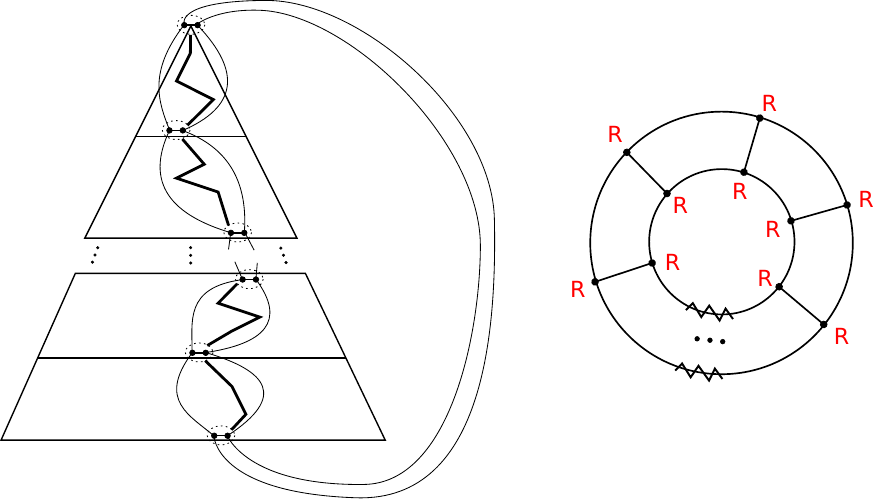}
    \end{center}
    \caption{Super-edge gadget for general case and how to achieve contradiction.}
    \label{fig:super_edge_general}
  \end{figure}
  
  Last thing to do is to prove that there exist a valid coloring of our super-edge gadget. Again we begin with labeling all
  nodes in tree by $X$. We know from previous section how to handle connections between root and leafs -- we have to do the same
  with all cycle-inducing nodes and their first cycle-inducing descendants. This way we will not be able to form a monochromatic cycle
  that passes through two different nodes that are in $CI$. Note that at this point the gadget is correct only when value
  $\lfloor \frac{k-1}{2} \rfloor$ is odd. This is true because of the way we color nodes in $CI$: the coloring of nodes on level
  $4i$ force the coloring on nodes on level $4(i+1)$. This problem can be easily fixed by expanding tree another 4 levels and treating
  nodes at level $h-4$ as {\em dummy nodes}.

  We are left with the case when $k$ is an odd number. Note that the construction above is not working in this case, as we will not
  achieve a contradiction. The fix is as follows: we change connections between leafs and root. Choose one vertex of root node and
  connect all vertices in leafs to it. This creates triangles rather than squares and the {\em donut} now looks like someone has
  taken a bite, but we can now find a monochromatic cycle of length $k$ for all odd numbers (if we color root node in red).
  The valid coloring does not change.

  For sake of completness we count the number of edges and vertices in entire reduction:

  \begin{align*}
    |E(G_{\phi})| &= |E_c| + |E_v| + |E_s| \\
    &= km + 0 + \\
    &+ (15 \cdot 2^{4\lfloor \frac{k-1}{2} \rfloor} - 2^{4\lfloor \frac{k-1}{2} \rfloor +1} \\
    &+ 2 \cdot 2^{4\lfloor \frac{k-1}{2} \rfloor} + 32 \sum_{i=0}^{\lfloor \frac{k-1}{2} \rfloor}2^{4i})(km+n)
  \end{align*}

  \begin{align*}
    |V(G_{\phi})| &= |V_c| + |V_v| + |V_s| \\
    &= km + 2n + (2(2^{4\lfloor \frac{k-1}{2} \rfloor + 1})-2)(km+n)
  \end{align*}

  Thus, we have proved theorem \ref{thm:main}.
  
  \section{Conclusions}
  
  We have shown that using symmetry, one can conceive many interesting combinatorial structures and in
  graph theory there is nothing more symmetric and regular than a clique. The obvious question is: can
  we make the reduction smaller? We have proved that string gadget from section 3 can be used as a tool
  to greatly decrease the number of edges and vertices used, but we do not know if the same can be said about
  tree gadget from sections 4 and 5.
  

\end{multicols}

\end{document}